%% file: lipics-main.tex
\documentclass[a4paper,UKenglish,cleveref, autoref, thm-restate]{lipics-v2021}
\usepackage[ruled,vlined]{algorithm2e}
\graphicspath{{figures/}}

\newcommand{\Poly}{{\mathsf{P}}}
\newcommand{\NP}{{\mathsf{NP}}}
\newcommand{\NC}{{\mathsf{NC}}}
\newcommand{\LS}{{\mathsf{LOGSPACE}}}

\renewcommand{\O}{\mathcal{O}}

\DeclareMathOperator{\depth}{\mathsf{depth}}
\DeclareMathOperator{\size}{\mathsf{size}}

\newlength{\LowLen}
\AtBeginDocument{\settoheight{\LowLen}{a}}
\DeclareMathOperator{\band}{\resizebox{!}{1.4\LowLen}{$\mathsf{AND}$}}
\DeclareMathOperator{\bor}{\resizebox{!}{1.4\LowLen}{$\mathsf{OR}$}}
\DeclareMathOperator{\bnot}{\resizebox{!}{1.4\LowLen}{$\mathsf{NOT}$}}

\DeclareMathOperator{\btrue}{\resizebox{!}{1.4\LowLen}{$\mathsf{TRUE}$}}
\DeclareMathOperator{\bfalse}{\resizebox{!}{1.4\LowLen}{$\mathsf{FALSE}$}}
\DeclareMathOperator{\copyg}{\resizebox{!}{1.4\LowLen}{$\mathsf{COPY}$}}
\DeclareMathOperator{\dupl}{\resizebox{!}{1.4\LowLen}{$\mathsf{DUPLICATION}$}}

\DeclareMathOperator{\smallband}{\resizebox{!}{.8\LowLen}{$\mathsf{AND}$}}
\DeclareMathOperator{\smallbnot}{\resizebox{!}{.8\LowLen}{$\mathsf{NOT}$}}
\DeclareMathOperator{\smallcopy}{\resizebox{!}{.8\LowLen}{$\mathsf{COPY}$}}
\DeclareMathOperator{\smalldupl}{\resizebox{!}{.8\LowLen}{$\mathsf{DUPLICATION}$}}

\newcommand{\LCP}{\mathcal{LCP}}
\newcommand{\LEP}{\mathcal{LEP}}
\newcommand{\CVP}{\mathcal{CVP}}

\newcommand{\decisionpb}[4]{\bigskip\noindent\fbox{\parbox{\textwidth}{
{#1} ({#2})\\
{\bf Input:} #3\\
{\bf Question:} #4
}}\bigskip}

\hideLIPIcs  

\title{Is Graph Local Complementation Inherently Sequential?} 

\author{Pablo {Concha-Vega}}{Aix-Marseille Universit{\'e} Toulon, LIS, CNRS UMR 7020, Marseille, France \and \url{https://pageperso.lis-lab.fr/~pablo.concha-vega/} }{pablo.concha-vega@lis-lab.fr}{https://orcid.org/0009-0001-2419-1687}{}

\authorrunning{P. Concha-Vega} 

\Copyright{Pablo Concha-Vega} 



\category{} 

\relatedversion{} 

\acknowledgements{}

\nolinenumbers 

\EventEditors{John Q. Open and Joan R. Access}
\EventNoEds{2}
\EventLongTitle{42nd Conference on Very Important Topics (CVIT 2016)}
\EventShortTitle{CVIT 2016}
\EventAcronym{CVIT}
\EventYear{2016}
\EventDate{December 24--27, 2016}
\EventLocation{Little Whinging, United Kingdom}
\EventLogo{}
\SeriesVolume{42}
\ArticleNo{23}

\begin{document}

\maketitle

\input{content}


\bibliography{refs}

\end{document}

%% file: content.tex
\begin{abstract}
  Local complementation of a graph $G$ on vertex $v$ is an operation that results
  in a new graph $G*v$, where the neighborhood of $v$ is complemented.
  Two graph are locally equivalent if on can be reached from the other one
  through local complementation.

  It was previously established that recognizing locally equivalent graphs
  can be done in $\O(n^4)$ time. We sharpen this result by proving it can
  be decided in $\O(\log^2(n))$ parallel time with $n^{\O(1)}$ processors.

  As a second contribution, we introduce the Local Complementation Problem,
  a decision problem that captures the complexity of applying a sequence of local
  complementations. Given a graph $G$, a sequence of vertices
  $s$, and a pair of vertices $u,v$, the problem asks whether the
  edge $(u,v)$ is present in the graph obtained after applying local complementations
  according to $s$. Regardless it simplicity, it is proven to be $\Poly$-complete,
  therefore it is unlikely to be efficiently parallelizable. 


  Finally, it is conjectured that Local Complementation Problem remains
  $\Poly$-complete when restricted to circle graphs.
\end{abstract}

\section{Introduction}

Local complementation of an undirected graph $G$ at a vertex $v$ is a local
operation that produces a new graph $G * v$, where the
induced graph on the neighborhood of $v$ is replaced by its complement and
the adjacency relations with the rest of the graph remain unchanged (see
Figure~\ref{fig:lc_ex}). This operation has been widely studied in various
contexts, including graph theory, coding theory, and quantum
computing~\cite{ehrenfeucht04,bouchet94,bouchet87,bouchet91,danielsen08,joo11,hahn19,dahlberg20}.

\begin{figure}
  \centering
  \begin{tabular}{ccc}
  \includegraphics[scale=1]{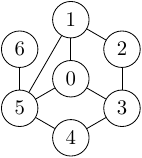}
  & &
  \includegraphics[scale=1]{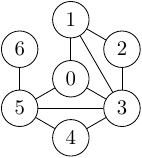}\\
  $G$ & & $G*0$
  \end{tabular}
  \caption{
    Example of local complementation. The operation is applied to vertex $0$,
    so the subgraph induced by its neighbors (vertices $1$, $3$, and $5$)
    is complemented in the resulting graph.
  }
  \label{fig:lc_ex}
\end{figure}

Two graphs are \textit{locally equivalent} if one can be obtained from the other
through a sequence of local complementations. Determining whether
two graphs are locally equivalent can be done in $\O(n^4)$ time~\cite{bouchet91}.
The first contribution of this work is to improve this bound by showing
that the problem belongs to $\NC^2$.

We then introduce the \emph{Local Complementation Problem} ($\LCP$),
a decision problem that captures the complexity of computing the effect of
successive local complementations. Given an undirected graph $G$, a sequence of
vertices $s$, and two vertices $u,v$, the problem asks whether the edge $(u,v)$ is
present in $G * s$, where $G * s$ denotes the graph obtained after applying local
complementation sequentially according to $s$.

The $\LCP$ can be seen as a minimal form of transformation problems
under local complementation: the sequence is fixed (i.e., there is
no need to choose which vertices to operate on), and the goal is not to reach
a target graph (possibly up to isomorphism), but simply to decide the presence or absence
of a single edge. Despite its simplicity, we prove its $\Poly$-completeness,
even though deciding local equivalence lies in $\NC$.

These results present a dichotomy: tracking the effect of a fixed sequence of local
complementations on a single edge is hard to parallelize, while computing the
global equivalence between two graphs is easy to compute. This highlights that
the computational nature of local complementation depends on whether one studies
its local action or its global effect.

\subsection{Previous work}

A graph $H$ is a \emph{vertex-minor}
of a graph $G$ if it is an induced subgraph of a graph locally equivalent to $G$.
Naturally, the operation of local complementation has been extensively studied
within the framework of vertex-minors theory. A comprehensive compendium of results
on this topic can be found in~\cite{kim24}.

A fundamental result on vertex-minors theory is the characterization of circle
graphs via forbidden vertex-minors~\cite{bouchet94}. A circle graph is defined by
a chord diagram where vertices correspond to chords and two vertices are adjacent
if and only if their corresponding chords intersect (see Figure~\ref{fig:circle_ex}
for an example). The theorem states that a graph is a circle graph if and only if
none of its vertex-minors is isomorphic to any of the graphs shown in Figure~\ref{fig:circle_forbidden}.
Furthermore, circle graphs are recognizable in polynomial time~\cite{bouchet87,gabor89}.

As mentioned in the introduction, it is known that recognizing a pair of
locally equivalent graphs can be done in $\O(n^4)$ time~\cite{bouchet91}.
Interestingly, deciding whether a graph is a vertex-minor of another graph or
isomorphic to a vertex-minor of another graph is an $\NP$-complete
problem~\cite{dahlberg22}. Previously, the same authors demonstrated that counting the
number of graphs locally equivalent to a given graph is $\#\Poly$-complete, even
when restricted to circle graphs~\cite{dahlberg20}.

Beyond classical graph theory, local complementation also plays a significant role
in quantum computing, particularly in the study of \emph{graph states}. A graph
state is a type of multi-qubit state that can be represented with a graph. In this
setting, local Clifford operations on graph states can be fully described using
local complementation~\cite{van04}. As a result, local complementation has been
widely studied in quantum computing as it represents a powerful tool for quantum
computing theorists. Notable works on this topic include~\cite{adcock20,hahn19,van04e}.

\begin{figure}
  \centering
  \raisebox{-0.5\height}{\includegraphics{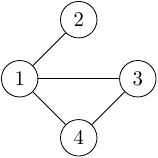}}
  \hspace*{.2in}
  \quad
  \raisebox{-0.5\height}{\includegraphics{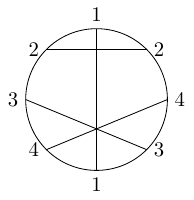}}
  \caption{Example of a circle graph and its chord diagram.}
  \label{fig:circle_ex}
\end{figure}

\begin{figure}
  \centering
  \includegraphics[scale=1]{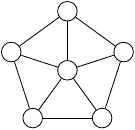}
  \quad
  \includegraphics[scale=1]{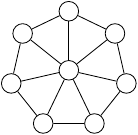}
  \quad
  \includegraphics[scale=1]{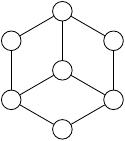}
  \caption{Forbidden vertex-minor graphs for the class of circle graphs.}
  \label{fig:circle_forbidden}
\end{figure}

\section{Preliminaries}

In this article, all graphs are simple.
Let $H$ be a graph. The set of vertices of $H$ is denoted $V(H)$, and the set
of edges by $E(H) \subseteq V(H) \times V(H)$ ($\subseteq {V(H) \choose 2}$
for undirected graphs).

let $G$ be an undirected graph. The \textit{local complementation} of $G$ at a
vertex $v$ is an operation resulting in a graph $G * v$ obtained by complementing the
adjacencies of all the neighbors of vertex $v$. In other words, two vertices
$u,u'\in V(G)$ are adjacent in $G * v$ if and only if exactly one of the following holds:

\begin{itemize}
  \item $(u,u') \in E(G)$
  \item $(u,v), (u',v) \in E(G)$
\end{itemize}

Applying local complementation to a graph $G$ over a sequence of vertices
$s = (v_1,\dots,v_k)$ is inductively defined as
$G*s = G*(v_1, \dots ,v_k) = G*v_1*(v_2, \dots ,v_k)$ and $G * \epsilon = G$.

\subsection{Boolean Circuits and $\NC$}

A \textit{Boolean circuit} $C$ is a labeled directed acyclic graph in which
edges carry unidirectional logical signals and the vertices compute elementary logical
functions. The entire graph computes a Boolean function from the inputs to the outputs
in a natural way. Every vertex $v \in V(C)$ has a type from $\{I, \band, \bor,
\bnot\}$. A vertex $v$ with type $I$ has in-degree $0$ and is called an \textit{input}.
The inputs of $C$ are given by a tuple $(x_1, \dots, x_n)$ of different vertices.
When computing a circuit, every input vertex has to be assigned a Boolean
value, these values are denoted by $(b_1, \dots, b_n) \in \{\bfalse, \btrue\}^n$.
A vertex $v$ with out-degree $0$ is called an \textit{output}.
The outputs of $C$ are given by a tuple $(y_1, \dots, y_m)$ of different vertices.
A vertex that is neither input nor output is called a \textit{gate}.
The size of $C$, denoted by $\size(C)$, is the number of vertices in $C$
The depth of $C$, denoted by $\depth(C)$, is the length of the longest path in $C$
from an input to an output.

In this work, adjacency matrices are regarded as the standard encoding of graphs,
for both directed and undirected. The \textit{adjacency matrix} of a graph $G$
is given by $A(G) = (a_{ij})$ where:

$$ a_{ij} =
  \begin{cases}
    1, & \text{if } (v_i,v_j) \in E(G)\\
    0, & \text{otherwise}  
  \end{cases}
$$

Given a circuit $C$, its encoding consists of two structures: the adjacency matrix
$A(C)$ of the underlying graph, and A sequence $t(C) \in \{I, \band, \bor,
\bnot\}^{\size(C)}$ storing the type of each vertex. The encoding of $C$ is
denoted by $\overline{C} = (A(C), t(C))$.

The class $\NC$ is defined in terms of Boolean circuits.
A problem is in $\NC$ if it can be solved by a family of Boolean circuits of
polylogarithmic depth and polynomial size. Equivalently, $\NC$ is
the class of problems solvable in polylogarithmic parallel time with a
polynomial number of processors by a PRAM (see \cite{greenlaw95} for a detailed definition).
For each $i \ge 1$, the subclass $\NC^i \subseteq \NC$
contains the problems decidables in $\O(\log^i n)$ parallel time. Furthermore:

$$
  \NC^1 \subseteq \LS \subseteq \NC^2 \subseteq \dots \subseteq \NC^i \subseteq \dots \subseteq NC \subseteq \Poly
$$

\subsection{Problems of Interest}

\subsubsection{The Circuit Value Problem}

The Circuit Value Problem ($\CVP$) is a widely studied decision problem
in $\Poly$-completeness theory and is considered the $\Poly$-complete
analogue of the Satisfiability Problem in $\NP$-completeness theory.
The problem is defined as follows: Given an encoding $\overline{C}$ of
a Boolean circuit $C$, input values $b_1, \dots, b_n$ and a
designated output $y$, determine the value of $y$ on the given input.
Formally:

\decisionpb
{Circuit Value Problem}
{$\CVP$}
{An encoding $\overline{C}$ of a circuit $C$, input values
$(b_1, \dots, b_n)$, a designated output $y$.}
{Is output $y$ of $C$ $\btrue$ on input values $(b_1,\dots b_n)$?}

Analogously to Satisfiability, $\CVP$ has many variants that are
$\Poly$-complete. One of them is the Topologically Ordered $\CVP$.
A \textit{topological ordering} of a circuit is a numbering of its
vertices such that $u$ is less than $v$ for every (directed) edge $(u,v)$.
Similarly, restraining the type of gates only to $\{\band, \bnot\}$ also
leads to a $\Poly$-complete problem. Interestingly, combining both
restrictions, also leads to a $\Poly$-complete version of $\CVP$.
Moreover, an even more restricted version considering only circuits of
fan-out at most $2$, is also $\Poly$-complete~\cite{greenlaw95}.
Although these restricted versions of the $\CVP$ have their own names
in the literature, this latter restricted $\CVP$ version will henceforth
be referred to simply as $\CVP$.

\subsubsection{The Local Complementation Problem}

The Local Complementation Problem ($\LCP$) is defined as follows: given
a graph $G$, a sequence of vertices $(v_1, \dots, v_k)$, and a pair
of vertices $u,v \in V(G)$, decide whether the edge $(u,v)$ belongs
to the graph obtained after applying the local complementation operation
on $G$ at each vertex in the given sequence in order. Formally:

\decisionpb
{Local Complementation Problem}
{$\LCP$}
{An undirected graph $G$, a sequence of vertices $(v_1, \dots, v_k)$ of $G$,
 a pair of vertices $u,v \in V(G)$.}
{Does $(u,v)$ belong to the graph $G * (v_1,\dots,v_k)$?} 

Note that the $\LCP$ is clearly in $\Poly$, as it can be solved by
a naive algorithm that iteratively applies the sequence of local
complementation to $G$.

\subsubsection{The Local Equivalence Problem}

To keep the presentation uniform, we introduce the following decision problem
concerning local equivalence:

\decisionpb
{Local Equivalence Problem}
{$\LEP$}
{Two undirected graphs $G_1$ and $G_2$ over the same vertex set $V=\{1,\dots,n\}$.}
{Is $G_1$ locally equivalent $G_2$?}

\section{$\LEP$ is in $\NC^2$}

In this section we sharpen the result of Bouchet~\cite{bouchet91} by proving
that deciding whether two graphs are locally equivalent is an $\NC^2$ problem.
First, a short summary of the relevant background is provided.

Let $G_1$ and $G_2$ be two undirected graphs over the same vertex-set
$V=\{1,\dots,n\}$. For $v,w,i \in V$, consider the following variables
over $\mathbb{F}_2$:

\begin{align*}
  \alpha_i^{vw} &= 1 \iff (i,v) \in E(G_1)\ \text{and}\ (i,w) \in E(G_2)\\
  \beta_i^{vw} &= 1 \iff (i,v) \in E(G_1)\ \text{and}\ i=w\\
  \gamma_i^{vw} &= 1 \iff i=v\ \text{and}\ (i,w) \in E(G_2)\\
  \delta_i^{vw} &= 1 \iff i=v=w
\end{align*}

\begin{theorem}[From~\texorpdfstring{\cite[Section~4]{bouchet91}}{}]
  Two graphs $G_1, G_2$ are locally equivalent if and only if the following
  systems of equations, with variables $X_i, Y_i, Z_i, T_i \in \mathbb{F}_2$,
  has a solution:

  \begin{flalign}
    \sum_{i=1}^n \left( \alpha_i^{vw}X_i + \beta_i^{vw}Y_i + \gamma_i^{vw}Z_i + \delta_i^{vw}T_i\right) = 0
    && \hfill \text{for all } v,w\in V && \label{eq:linear}\\
    X_i T_i + Y_i Z_i = 1 && \hfill \text{for all } i\in V && \label{eq:nonlinear}
  \end{flalign}
\end{theorem}

The results from~\cite[Theorems~1 and~5]{borodin82} and~\cite{mulmuley86} show that
there exists an $\NC^2$ algorithm for finding a basis for the nullspace
of a given matrix with coefficients over an arbitrary field.
Therefore, we can directly apply this result to
the linear system~(\ref{eq:linear}).
The solutions to~(\ref{eq:linear}) are a subspace $\mathcal{S}$ of $\mathbb{F}_2^{4n}$.
Let us denote a basis of $\mathcal{S}$ by $B$.
Note that $\dim(B) = \O(n)$ since there are $4n$ variables. 
A naive algorithm for solving the full system composed of~(\ref{eq:linear})
and~(\ref{eq:nonlinear}) would consist of evaluating
all the linear combinations of $B$ over~(\ref{eq:nonlinear}), however
the number of combinations is potentially exponential.
Fortunately, Bouchet also proved the following two results.

\begin{lemma}[From~\texorpdfstring{\cite[Lemma~4.3]{bouchet91}}{}]
  If the system of equations consisting of~(\ref{eq:linear}) and~(\ref{eq:nonlinear})
  has any solution and $\dim(\mathcal{S}) > 4$, then there exists an affine subspace
  $\mathcal{A}$ of $\mathcal{S}$ such that $\dim(\mathcal{S}) - \dim(\mathcal{A}) \leq 2$
  and every $a \in \mathcal{A}$ is a solution to~(\ref{eq:linear}) and~(\ref{eq:nonlinear}). 
\end{lemma}

\begin{lemma}[From~\texorpdfstring{\cite[Lemma~4.4]{bouchet91}}{}]
  For every basis $B$ of a vector space $\mathcal{S}$ over $\mathbb{F}_2$ and
  every affine subspace $\mathcal{A}$ of $\mathcal{S}$ such that
  $\dim(\mathcal{S}) - \dim(\mathcal{A}) \leq q$, there exists a vector
  $a \in \mathcal{A}$ which is the sum of at most $q$ vectors of $B$.
\end{lemma}

Consequently, if the full system has a solution, there always exists
one that is a linear combination of at most two vectors of $B$
(if $\dim(\mathcal{S}) \leq 4$ one can simply check the $\leq 16$ possible candidates).
Determining whether a vector $b_k \in B$ satisfies~(\ref{eq:nonlinear})
can be decided in $\NC$. Let assume that $b_k$ is shaped as

$$
  b_k = (X_1^k, Y_1^k, Z_1^k, T_1^k, \dots X_n^k, Y_n^k, Z_n^k, T_n^k)
$$

Then, using $n$ processors one can verify~(\ref{eq:nonlinear}) for $b_k$,
as described in Algorithm~\ref{alg:check_vector}.
Each verification produces a Boolean value that is stored in an array.
This array is the input of a prefix-$\band$, which can be performed in
$\NC^2$~\cite{jaja92}, and outputs $1$ if and only if the vector is
a solution to~(\ref{eq:nonlinear}).
Since each local verification requires $\O(1)$ parallel time and the
prefix-$\band$ runs in $\O(\log n)$ time, the whole procedure decides
whether $b_k$ satisfies~(\ref{eq:nonlinear}) in $\O(\log n)$ parallel time.
Running this algorithm simultaneously for all $\O(n)$ vectors of $B$ requires
$\O(n^2)$ processors.

\input{alg_combination}

The next step is to check all pairwise combinations of the basis $B$.
Using a similar technique, each processor is assigned two corresponding
blocks of four coordinates (one from each of two vectors).
The processor computes their componentwise sum~(modulo $2$) and then verifies
whether the resulting block satisfies condition~(\ref{eq:nonlinear}).
Since there are $\O(n^2)$ possible combinations, the algorithm now requires
$\O(n^3)$ processors. As before, the outcomes are collected in an array of
Boolean values, and a prefix-$\band$ determines whether a given pair of
vectors yields a valid solution. 

Since computing $B$ is the computationally hardest part overall,
the following result follows:

\begin{theorem}
  $\LEP$ is in $\NC^2$.
\end{theorem}

\section{$\Poly$-completeness of $\LCP$}

This section studies the $\Poly$-completeness of $\LCP$.
This is proven by reducing the $\CVP$ to $\LCP$.
This is achieved by constructing a series of
\emph{gadgets} to simulate circuits.
From now on, when reducing from the previously defined restricted version
of the $\CVP$, it is assumed that circuits have a topological ordering,
are composed only of $\band$ and $\bnot$, and have fan-out at most $2$.

\subsection{Circuit Gadgets}

Although the aim is to simulate circuits, the nature of the proof needs not
only $\band$ and $\bnot$ gadgets but also the inclusion of a $\copyg$ gadget
and a $\dupl$ gadget. Every gadget is composed by a graph, a sequence of
vertices, inputs and outputs. The main idea is that a gadget must compute
its task by orderly applying the local complementation operation over the
vertices of its sequence. It is important to note that, in order to prove the
main result, the sequence for each gadget must remain constant; that is,
it cannot depend on the input. Instead, the same
sequence must function correctly for any possible input to the gadget.

Inputs and outputs in the circuit are represented by the existence or absence
of predefined edges in the graph, therefore, for every input $x_i$ of a
circuit $C$, a pair of vertices $v_i, v'_i$ is created and they are connected
by an edge if and only if its Boolean value $b_i = \btrue$.

In the same fashion, $\copyg$ of a $\btrue$ variable
consists of ``turning on'' another predefined edge through a sequence of local
complementations. On the other hand, if the variable has a $\bfalse$ value,
the same sequence of local complementations will not turn on the designated
edge. In particular, this gadget is a quadruple
$(G_{\smallcopy}, s_{\smallcopy}, in, out)$ such that:

\begin{itemize}
  \item $V(G_{\smallcopy}) = \{0,1,2,3\}$,
  \item $E(G_{\smallcopy}) = \{(0,2), (1,3)\}$,
  \item $s_{\smallcopy} = (0,1,0)$,
  \item $in = (0,1)$ represents the input of the gadget,
  \item $out = (2,3)$ represents the output of the gadget.
\end{itemize}

The functioning of this gadget is depicted in Figure~\ref{fig:copy}.

\begin{figure}
  \centering
  \begin{subfigure}[b]{.2\textwidth}
    \centering
    \includegraphics[scale = .6]{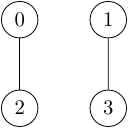}
    \subcaption{}
  \end{subfigure}%
  \begin{subfigure}[b]{.2\textwidth}
    \centering
    \includegraphics[scale = .6]{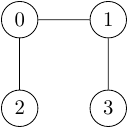}
    \subcaption{}
  \end{subfigure}%
  \begin{subfigure}[b]{.2\textwidth}
    \centering
    \includegraphics[scale = .6]{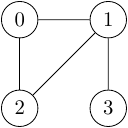}
    \subcaption{}
  \end{subfigure}%
  \begin{subfigure}[b]{.2\textwidth}
    \centering
    \includegraphics[scale = .6]{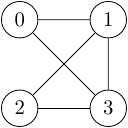}
    \subcaption{}
  \end{subfigure}%
  \begin{subfigure}[b]{.2\textwidth}
    \centering
    \includegraphics[scale = .6]{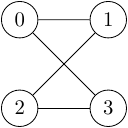}
    \subcaption{}
  \end{subfigure}
  \caption{
    Functioning of the $\copyg$ gadget. (a) Gadget initialized with a $\bfalse$
    input, i.e., the edge $(0,1)$ is not part of the graph.
    After applying the local complementation on the vertices given by
    $s_{\smallcopy}$ the graph remains invariant
    ($G_{\smallcopy} = G_{\smallcopy} * s_{\smallcopy}$).
    (b) Gadget initialized as $\btrue$. After applying the local
    complementations on the vertices $s_{\smallcopy}$ one obtains the graphs
    (c), (d) and (e). Note that the resulting graph has the edge $(2,3)$,
    which is the predefined output edge.
    }
  \label{fig:copy}
\end{figure}

The $\bnot$ gadget is now introduced. In order to negate a variable represented
by two vertices $v, v'$, it is enough to add a new vertex $u$ together with the
edges $(u,v)$ and $(u,v')$. Then the local complementation on vertex $u$ toggles
the adjacency relation of $v$ and $v'$. A few vertices and edges are added for
clarity. Therefore, the $\bnot$ gadget is a quadruple
$(G_{\smallbnot}, s_{\smallbnot}, in, out)$ such that:

\begin{itemize}
  \item $V(G_{\smallbnot}) = \{0,1,2,3,4\}$,
  \item $E(G_{\smallbnot}) = \{(0,2), (1,3), (2,4), (3,4)\}$,
  \item $s_{\smallbnot} = (0,1,0,4)$,
  \item $in = (0,1)$ represents the input of the gadget,
  \item $out = (2,3)$ represents the output of the gadget.
\end{itemize}

Note that the $\bnot$ gadget is a $\copyg$ gadget with an extra vertex
that computes the negation. Its functioning is depicted in
Figures~\ref{fig:not_f} and~\ref{fig:not_t}.

\begin{figure}
  \centering
  \begin{subfigure}[b]{.48\textwidth}
    \centering
    \includegraphics[scale = .6]{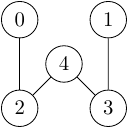}
    \subcaption{}
  \end{subfigure}%
  \begin{subfigure}[b]{.48\textwidth}
    \centering
    \includegraphics[scale = .6]{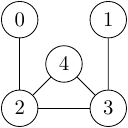}
    \subcaption{}
  \end{subfigure}%
  \caption{
    Functioning of the $\bnot$ gadget initialized with a $\bfalse$ input.
    (a) The gadget remains identical for the first three local complementations.
    (b) Graph obtained after the last operation.
  }
  \label{fig:not_f}
\end{figure}
\begin{figure}
  \centering
  \begin{subfigure}[b]{.2\textwidth}
    \centering
    \includegraphics[scale = .6]{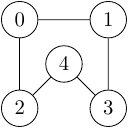}
    \subcaption{}
  \end{subfigure}%
  \begin{subfigure}[b]{.2\textwidth}
    \centering
    \includegraphics[scale = .6]{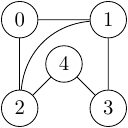}
    \subcaption{}
  \end{subfigure}%
  \begin{subfigure}[b]{.2\textwidth}
    \centering
    \includegraphics[scale = .6]{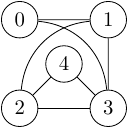}
    \subcaption{}
  \end{subfigure}%
  \begin{subfigure}[b]{.2\textwidth}
    \centering
    \includegraphics[scale = .6]{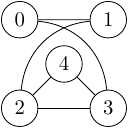}
    \subcaption{}
  \end{subfigure}%
  \begin{subfigure}[b]{.2\textwidth}
    \centering
    \includegraphics[scale = .6]{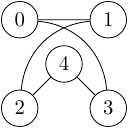}
    \subcaption{}
  \end{subfigure}%
  \caption{
    Functioning of the $\bnot$ gadget initialized with a $\btrue$ input.
  }
  \label{fig:not_t}
\end{figure}

The focus now shifts to the $\band$ gadget. The $\band$ gadget is the only
gadget with two input values, then it has to correctly compute the $\band$
function of all possible $4$ combinations of inputs. It consists of a
quintuple $(G_{\smallband}, s_{\smallband}, in, in', out)$ such that

\begin{itemize}
  \item $V(G_{\smallband}) = \{0,1,2,3,4,5,6\}$,
  \item $E(G_{\smallband}) = \{(0,4), (1,5), (2,6), (3,4)\}$,
  \item $s_{\smallband} = (1,2,0,3,4)$,
  \item $in = (0,1)$ represents the first input of the gadget,
  \item $in' = (2,3)$ represents the second input of the gadget,
  \item $out = (5,6)$ represents the output of the gadget.
\end{itemize}

The functioning of the $\band$ gadget is depicted in Figures~\ref{fig:and_f_t}
and \ref{fig:and_t_t}, for the cases in which $in = \btrue$, $in' = \bfalse$
and $in = in' = \btrue$, respectively. The case $in=in'=\bfalse$ is
straightforward as it only generates one new edge ($(0,3)$). The case
$in = \bfalse$, $in' = \btrue$ is symmetrical to $in = \btrue$, $in' = \bfalse$.

\begin{figure}
  \centering
  \begin{subfigure}[b]{.33\textwidth}
    \centering
    \includegraphics[scale = .5]{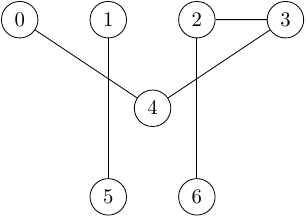}
    \subcaption{}
  \end{subfigure}%
  \begin{subfigure}[b]{.33\textwidth}
    \centering
    \includegraphics[scale = .5]{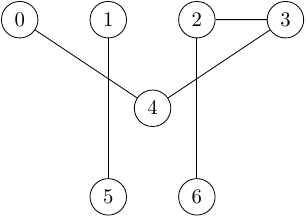}
    \subcaption{}
  \end{subfigure}%
  \begin{subfigure}[b]{.33\textwidth}
    \centering
    \includegraphics[scale = .5]{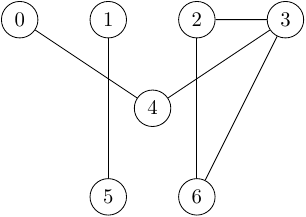}
    \subcaption{}
  \end{subfigure}\newline
  \begin{subfigure}[b]{.33\textwidth}
    \centering
    \includegraphics[scale = .5]{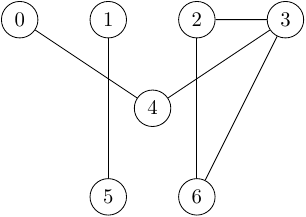}
    \subcaption{}
  \end{subfigure}%
  \begin{subfigure}[b]{.33\textwidth}
    \centering
    \includegraphics[scale = .5]{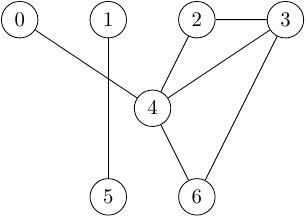}
    \subcaption{}
  \end{subfigure}%
  \begin{subfigure}[b]{.33\textwidth}
    \centering
    \includegraphics[scale = .5]{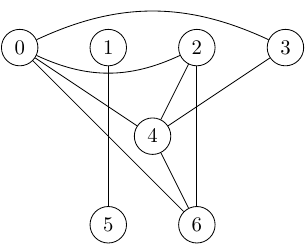}
    \subcaption{}
  \end{subfigure}%
  \caption{Functioning of the $\band$ gadget with $in = \bfalse$, $ in' = \btrue$.}
  \label{fig:and_f_t}
\end{figure}

\begin{figure}
  \centering
  \begin{subfigure}[b]{.33\textwidth}
    \centering
    \includegraphics[scale = .5]{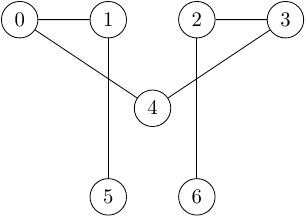}
    \subcaption{}
  \end{subfigure}%
  \begin{subfigure}[b]{.33\textwidth}
    \centering
    \includegraphics[scale = .5]{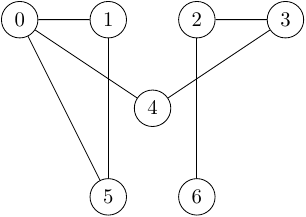}
    \subcaption{}
  \end{subfigure}%
  \begin{subfigure}[b]{.33\textwidth}
    \centering
    \includegraphics[scale = .5]{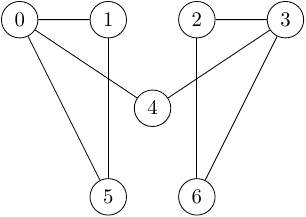}
    \subcaption{}
  \end{subfigure}\newline
  \begin{subfigure}[b]{.33\textwidth}
    \centering
    \includegraphics[scale = .5]{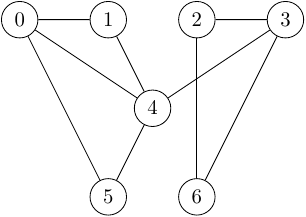}
    \subcaption{}
  \end{subfigure}%
  \begin{subfigure}[b]{.33\textwidth}
    \centering
    \includegraphics[scale = .5]{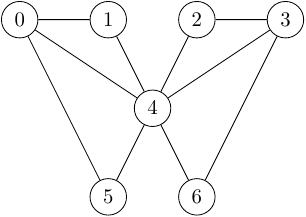}
    \subcaption{}
  \end{subfigure}%
  \begin{subfigure}[b]{.33\textwidth}
    \centering
    \includegraphics[scale = .5]{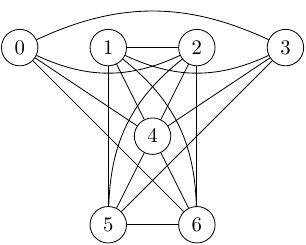}
    \subcaption{}
  \end{subfigure}%
  \caption{Functioning of the $\band$ gadget with $in = in' = \btrue$.}
  \label{fig:and_t_t}
\end{figure}

Finally, the $\dupl$ gadget, as its name suggests, duplicates a variable,
therefore, it has one input and two outputs. Formally, it is defined as
a quintuple $(G_{\smalldupl}, s_{\smalldupl}, in, out, out')$ such that:

\begin{itemize}
  \setlength\itemsep{0em}
  \item $V(G_{\smalldupl}) = \{0,1,2,3,4,5,6,7\}$,
  \item $E(G_{\smalldupl}) = \{(0,2), (0,4), (1,3), (1,5), (2,5), (2,6), (3,4), (3,7)\}$,
  \item $s_{\smalldupl} = (0,1,0,2,3,2,0,3,0)$,
  \item $in = (0,1)$ represents the input of the gadget,
  \item $out = (4,5)$ represents the first output of the gadget,
  \item $out' = (6,7)$ represents the second output of the gadget.
\end{itemize}

The functioning of the $\dupl$ gadget is depicted in
Figures~\ref{fig:duplication_f} and~\ref{fig:duplication_t}.
According to the previous definition, the last $0$ in $s_{\smalldupl}$
might seem redundant, since after $8$ steps the output variables are
already connected if and only if the corresponding Boolean value is $\btrue$.
However, this extra step is added to ensure an additional important
property: vertices of different outputs are not connected.
This is crucial because when gadgets are connected, local complementations
are applied on input vertices, and the presence of such edges would cause
conflicts and invalidate the construction.

\begin{figure}
  \centering
  \begin{subfigure}[b]{.33\textwidth}
    \centering
    \includegraphics[scale=.5]{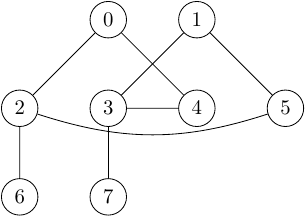}
    \subcaption{}
  \end{subfigure}%
  \begin{subfigure}[b]{.33\textwidth}
    \centering
    \includegraphics[scale=.5]{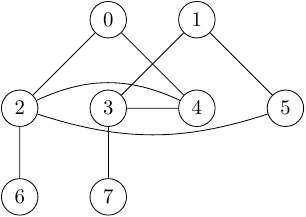}
    \subcaption{}
  \end{subfigure}%
  \begin{subfigure}[b]{.33\textwidth}
    \centering
    \includegraphics[scale=.5]{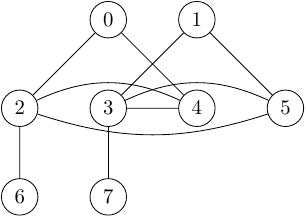}
    \subcaption{}
  \end{subfigure}\newline
  \begin{subfigure}[b]{.33\textwidth}
    \centering
    \includegraphics[scale=.5]{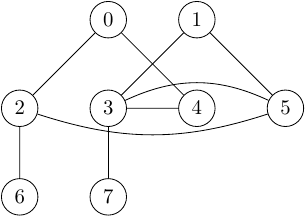}
    \subcaption{}
  \end{subfigure}%
  \begin{subfigure}[b]{.33\textwidth}
    \centering
    \includegraphics[scale=.5]{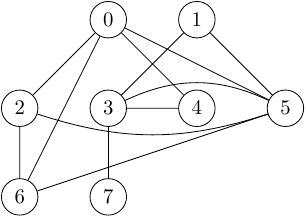}
    \subcaption{}
  \end{subfigure}%
  \begin{subfigure}[b]{.33\textwidth}
    \centering
    \includegraphics[scale=.5]{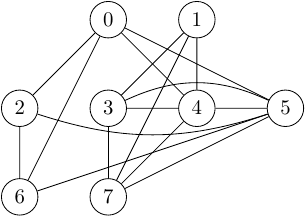}
    \subcaption{}
  \end{subfigure}\newline
  \begin{subfigure}[b]{.25\textwidth}
    \centering
    \includegraphics[scale=.5]{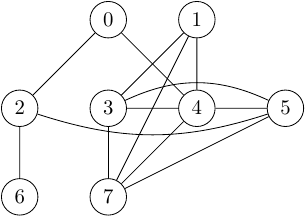}
    \subcaption{}
  \end{subfigure}%
  \begin{subfigure}[b]{.25\textwidth}
    \centering
    \includegraphics[scale=.5]{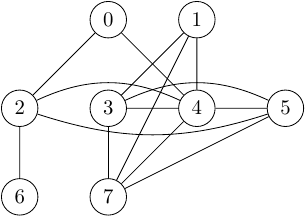}
    \subcaption{}
  \end{subfigure}%
  \begin{subfigure}[b]{.25\textwidth}
    \centering
    \includegraphics[scale=.5]{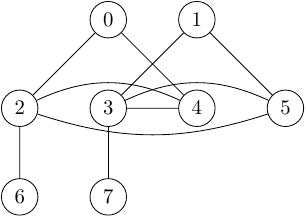}
    \subcaption{}
  \end{subfigure}%
  \begin{subfigure}[b]{.25\textwidth}
    \centering
    \includegraphics[scale=.5]{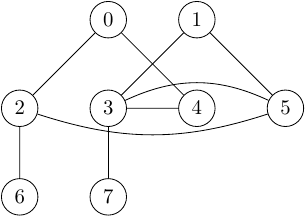}
    \subcaption{}
  \end{subfigure}%
  \caption{Functioning of the $\dupl$ gadget with input $\bfalse$.}
  \label{fig:duplication_f}
\end{figure}

\begin{figure}
  \centering
  \begin{subfigure}[b]{.33\textwidth}
    \centering
    \includegraphics[scale=.5]{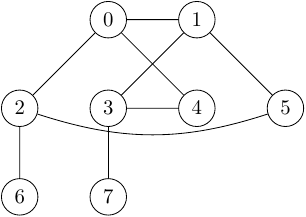}
    \subcaption{}
  \end{subfigure}%
  \begin{subfigure}[b]{.33\textwidth}
    \centering
    \includegraphics[scale=.5]{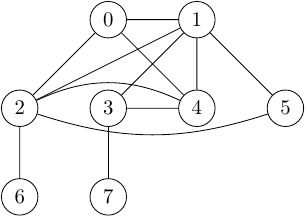}
    \subcaption{}
  \end{subfigure}%
  \begin{subfigure}[b]{.33\textwidth}
    \centering
    \includegraphics[scale=.5]{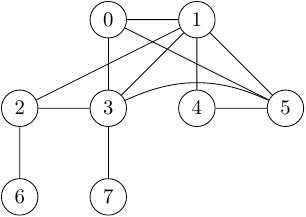}
    \subcaption{}
  \end{subfigure}\newline
  \begin{subfigure}[b]{.33\textwidth}
    \centering
    \includegraphics[scale=.5]{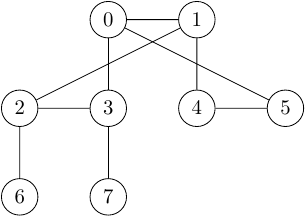}
    \subcaption{}
  \end{subfigure}%
  \begin{subfigure}[b]{.33\textwidth}
    \centering
    \includegraphics[scale=.5]{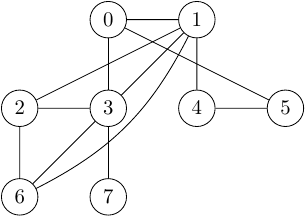}
    \subcaption{}
  \end{subfigure}%
  \begin{subfigure}[b]{.33\textwidth}
    \centering
    \includegraphics[scale=.5]{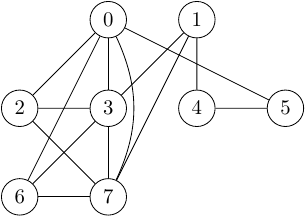}
    \subcaption{}
  \end{subfigure}\newline
  \begin{subfigure}[b]{.25\textwidth}
    \centering
    \includegraphics[scale=.5]{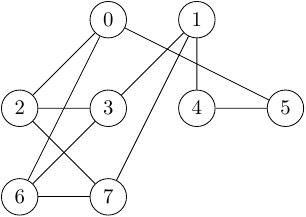}
    \subcaption{}
  \end{subfigure}%
  \begin{subfigure}[b]{.25\textwidth}
    \centering
    \includegraphics[scale=.5]{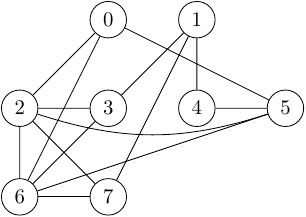}
    \subcaption{}
  \end{subfigure}%
  \begin{subfigure}[b]{.25\textwidth}
    \centering
    \includegraphics[scale=.5]{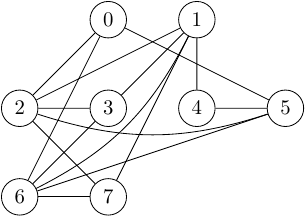}
    \subcaption{}
  \end{subfigure}%
  \begin{subfigure}[b]{.25\textwidth}
    \centering
    \includegraphics[scale=.5]{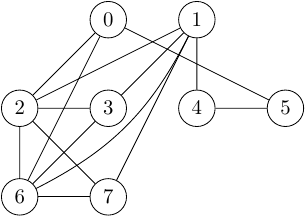}
    \subcaption{}
  \end{subfigure}%
  \caption{Functioning of the $\dupl$ gadget with input $\btrue$.}
  \label{fig:duplication_t}
\end{figure}

\subsection{Combining Gadgets}\label{subs:combining}

With the toolkit of gadgets now defined, the next step is to describe
how to combine them to simulate circuits.

A first step is describing how to combine two gadgets together.
In order to accomplish the latter, one must replace the output $(u,u')$
of a gadget with the input $(v,v')$ of a second gadget by collapsing
$u$ (resp. $u'$) with $v$ (resp. $v'$) into one vertex $p$ (resp
$p'$) while keeping all their edges, i.e., for every edge $(s,u)$
and $(t,v)$ (resp. $(s', u')$ and $(t',v')$) the edges $(s,p)$ and
$(t,p)$ (resp. $(s',p')$ and $(t',p')$) have to be created. Also,
the sequence $s_1$ of the first gadget has to be concatenated with
the sequence $s_2$ of the second gadget, resulting in $s_1s_2$.
An example is depicted in Figure~\ref{fig:combined_gadgets}.
Note that none of the gadgets consider applying local complementation
on their output vertices, therefore concatenating their sequences
does not leave to any kind of malfunctioning.

\begin{figure}
  \centering
  \includegraphics[scale = .5]{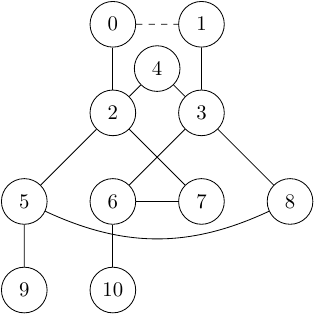}
  \caption{
    A $\bnot$ gadget and a $\dupl$ gadget combined. Local complementing
    on the sequence $(0,1,0,4,2,3,2,5,6,5,2,6,2)$ computes the duplication
    of the negation of the variable encoded in the adjacency relation
    of vertices $0$ and $1$.
  }
  \label{fig:combined_gadgets}
\end{figure}

\subsection{Building Circuits}

This subsection presents the construction of an arbitrary circuit by
combining the previously defined gadgets. Additionally, the correctness
of this construction is established. An example of this construction
is illustrated in Figure~\ref{fig:construction_ex}.

\begin{figure}
  \centering
  \includegraphics[scale = .8]{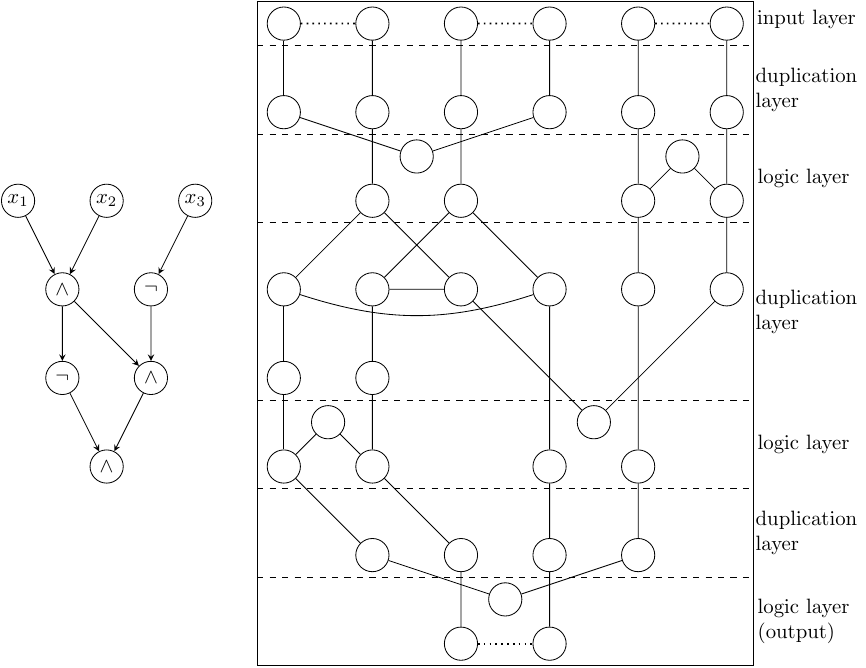}
  \caption{
    Left: example of a Boolean circuit with out-degree $2$. This circuit
    is composed only of $\band$ and $\bnot$ gates. Right: The construction
    of a graph that can simulate the circuit by applying local 
    complementation. The sequence corresponds to the concatenation of
    the sequences of the gadgets, layer by layer.
  }
  \label{fig:construction_ex}
\end{figure}

Given a circuit $C$ a graph-sequence structure (GSS) $(G,s,u,v)$ is constructed
where that $G$ is an undirected graph, $s$ a sequence of vertices of $G$
and $u,v \in V(G)$ represent the output of $C$. The construction is now presented.

The process starts with an empty graph-sequence structure.
First, for every input $x_i$ of a circuit $C$, a pair of
vertices $v_i, v'_i$ is created. The edge $(v_i, v'_i)$ is added if and
only if its associated Boolean value $b_i = \btrue$.
This structure is referred to as the \emph{input layer}.
After the input layer, duplication layers alternate
with logic layers, concluding with a final logic layer that contains
the circuit's output. Note that the sequence associated to the input
layer is the empty sequence, then nothing has to be added to $s$.

A \emph{duplication layer} is composed of $\dupl$ and $\copyg$ gadgets.
Depending on whether a Boolean value is used once or twice, one of these
two gadgets is selected and attached. Duplication layers are initially
placed after the input layer and subsequently after each logic layers.
Despite their placement, the construction process remains the same in
both cases. For every pair of vertices $v,v'$ representing a Boolean
variable in the previous layer, whether originating from the input
layer or from the output of a logic layer, a $\dupl$ or a $\copyg$ gadget
is attached. The sequences of all the gadgets used in a layer are
concatenated one after the other in no particular order and then
concatenated at the end of $s$.

On the other hand, a \emph{logic layer} is composed of $\band$ and
$\bnot$ gadgets. Similarly as before, the outputs of the previous
layer are the inputs of a logic layer as specified in the circuit $C$.
Again, the sequences of all the gadgets used in a layer are concatenated
one after the other in no particular order and then added to $s$. The
last layer of this construction is a logic layer that contains the
pair of vertices $u,v$ representing the output of $C$.

The correctness of this structure is now proven.

\begin{lemma}\label{lemma:graph_sequence_structure}
  Given a circuit $C$, input values $(b_1,\dots,b_n)$ and the GSS of
  $C$ as $(G,s,u,v)$, then $(u,v) \in G*s$ if and only if
  $C$ outputs $\btrue$ on inputs $x$.
\end{lemma}

\begin{proof}

  First of all, notice that the sequences associated to the gadgets
  exclude the vertices defined as gadget outputs. This ensures that
  during the simulation, applying local complementation on the
  vertices specified by $s$ within a particular layer never results
  in edges connecting to vertices in the subsequent layers.

  The proof proceeds by induction on the layers of the GSS representing
  the circuit. The base case establishes that the input layer combined
  with the first duplication and logic layers, correctly simulates the
  first layer of gates of $C$.

  For proving the base case, local complementation is applied
  to the vertices given by $s$ corresponding to the first three layers
  of the GSS. Since the input layer works simply as the representation
  of the inputs $x$, the process begins by local complementing the vertices
  specified by $s$ on the first duplication layer. This step preserves the
  behavior of the $\dupl$ and $\copyg$ gadgets in their original form,
  ensuring they work as intended.

  Attention then shifts to the first logic layer. Although the gadgets
  in this layer are no longer in their original form, due to the input
  vertices having additional edges, applying local complementation on the
  vertices specified by $s$ results in irrelevant edges with the
  preceding layers and in the desired edges from the original functioning
  of the gadgets. The edges connecting previous layers are irrelevant
  because those vertices will not undergo further local complementation.
  This behavior is present across all the subsequent layers.

  For the induction step, the assumption is made that the simulation works
  as expected up to the the $k$-th logic layer. The goal is to prove that
  the subsequent duplication and logic layers also behave as intended.

  First, consider the $k+1$-th duplication layer. By the induction
  hypothesis, the $k$-th logic layer correctly represent the outputs of
  the corresponding gates of $C$. The gadget inputs at duplication layer
  $k+1$ are these outputs. Similarly as before, applying
  local complementation as specified by $s$ transforms the $\copyg$
  and $\dupl$ gadgets as intended but with extra edges connecting
  the previous layers of the GSS. Therefore, the $k+1$-th duplication
  layer functions as intended.

  Finally, consider the $k+1$-th logic layer. Gadget inputs in this
  layer are the outputs of the $k+1$-th duplication layer. Therefore,
  applying local complementation on the vertices specified by $s$
  work as intended. Thus, the correctness of the construction extends
  to logic gate $k+1$, completing the induction step.
\end{proof}

\subsection{Computing the GSS of a Circuit in $\LS$}

Lemma~\ref{lemma:graph_sequence_structure} proves that simulating a
circuit using its GSS correctly computes its output. If the GSS of a
circuit can be computed in $\LS$, it implies that $\CVP$ reduces to
$\LCP$ under $\LS$ reductions. Consequently, proving this fact
demonstrates the main result of this article, showing that $\LCP$ is
$\Poly$-complete.

The GSS of a circuit can be indeed computed in $\LS$ due to its
structured and modular nature (the graph and the sequence are
composed of fixed-size objects). The graph is constructed by
``gluing'' predefined, fixed-size graphs corresponding to the
components of the circuits. This operation only requires keeping
track of the components and their connections, which can be achieved
in logarithmic space.

Similarly, the sequence associated with the GSS is formed by
concatenating the fixed sequences of the gadgets. Since this
concatenation is performed directly based on the circuit's structure
and it does not depend on global information, it can also be
managed in $\LS$. Finally, the two special vertices of the GSS
representing the output of the circuit are identified during the
construction with minimal usage of space. The following theorem is
then concluded.

\begin{theorem}
  $\LCP$ is $\Poly$-complete.
\end{theorem}

\subsection{Further Observations}\label{sub:further_observations}

An interesting variation of the $\LCP$ extends the allowed operations
to both local complementation and vertex deletions, making it the
``vertex-minor counterpart'' of the $\LCP$. In this case, given
a graph $G$, a sequence $s$ of pairs $(v,operation)$, where each pair
consists of a vertex $v$ and an operation (either local complementation
or vertex deletion), and a pair of special vertices $u,u'$, the problem
asks whether $u$ and $u'$ are connected in the graph obtained after
applying the sequence of operations from $s$. This variant remains
$\Poly$-complete, which directly follows from the $\Poly$-completeness
of the $\LCP$.

\textit{Subgraph complementation} is a local transformation also
studied in the literature (see~\cite{fomin20} for example).
A subgraph complement of a graph $G$ is
a graph obtained by complementing one of its induced subgraphs.
Thus, local complementation can be seen as a special case of subgraph
complementation. A variation of the $\LCP$ in which the sequence
of vertices is replaced by a sequence of vertex subsets to apply subgraph
complementation, would also be $\Poly$-complete.

Another noteworthy observation comes from restricting the $\LCP$ to complete
graphs and star graphs. The set of graphs
locally equivalent to a complete graph is quite restricted: these are exactly
the star graphs of the same size (i.e., with the same number of vertices), and
vice versa (see Figure~\ref{fig:complete_and_star}). As a result, to decide the
problem, one only needs to track whether the graph is a complete graph or a
star graph, and if it is a star graph, which vertex serves as the center of the
star. Since this information can be stored in a single integer, the problem can
be categorized in $\LS$. This shows that, for these specific graph classes,
the complexity of the problem is significantly reduced.

\begin{figure}
  \centering
  \includegraphics[scale=1]{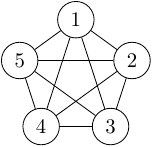}
  \quad
  \includegraphics[scale=1]{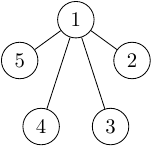}
  \caption{
    The complete graph $K_5$ and $K_5*1$, the star graph centered on
    vertex $1$.}
  \label{fig:complete_and_star}
\end{figure}

\section{Conclusions}

In this work the $\LCP$ is introduced and proved to be $\Poly$-complete.
This problem captures the complexity of the operation of local complementation
by asking whether the graph resulting from applying this operation following a
given sequence contains or not a specific edge. This gives valuable insights,
as calculating a single edge through a series of \emph{given} local complementations
is unlikely to be done with more efficiency unless $\Poly$ equals a lower
complexity class such as $\NC$ or $\LS$.

In contrast, for the closely related $\LEP$, it was established that the problem
can be solved in $\NC^2$, showing that determining whether two graphs are locally
equivalent admits an efficient parallel algorithm. Together, these results
delineate the frontier between problems on local complementation that are
parallelizable and those that remain inherently sequential.

In subsection~\ref{sub:further_observations}, it was observed that variations of the 
$\LCP$ present interesting challenges, with their computational complexity varying
significantly. These variations might involve considering different operations,
such as pivoting or switching,
instead of local complementation, or combinations of operations as mentioned earlier
(e.g., local complementation combined with vertex deletions). The complexity of
such variants can be influenced by restricting the number of times an operation
can be applied to a vertex or by limiting the graph class under consideration,
as seen in the case of complete graphs and star graphs.

A particularly interesting class of graphs are circle graphs, which serve as a
canonical example of a class closed under local complementation. It is worth
noting that all the gadgets presented in this paper are circle graphs, which is
not difficult to verify (a task left to the reader). I conjecture that the problem
remains $\Poly$-complete when restricted to circle graphs, though this remains an
open question.

The results of this paper also have implications in the context of graph states in
quantum computing. As local complementation plays a crucial role in the manipulation
of graph states, the $\Poly$-completeness of the $\LCP$ together with the parallel
tractability of $\LEP$ indicate that, while deciding Clifford equivalence is efficiently
parallelizable, predicting the effect of a prescribed sequence of operations on a
specific edge is computationally harder.

\section*{Acknowledgements}
I thank Kévin Perrot for helpful observations and discussions.

%% file: alg_combination.tex
\SetKwFor{ParFor}{for}{in parallel do}{end}
\begin{algorithm}[t]
\SetAlgoLined
\KwIn{A vector $b_k$ of basis $B$}
\KwOut{$1$ if $b_k$ is a solution to~(\ref{eq:nonlinear}), $0$ otherwise}
$\textsf{valid = [0,\dots,0]}$ \tcp*[r]{array of size $n$}
\ParFor{$1 \leq j \leq n$}{
    $index = 4(j-1) + 1$\\
    $X, Y, Z, T = b[index], b[index+1], b[index+2], b[index+3]$\\
    \If{$XT + YZ = 1$}{
        $\textsf{valid[j] = 1}$
    }
}
\Return AND($\textsf{valid}$) \tcp*[r]{prefix-$\band$}
\caption{Checks if a vector satisfies~(\ref{eq:nonlinear}).}
\label{alg:check_vector}
\end{algorithm}